\documentclass{llncs}
\usepackage[utf8]{inputenc}
\title{Automated Synthesis of Safe and Robust PID Controllers for Stochastic Hybrid Systems}

\author{
Fedor Shmarov\inst{1} \and 
 Nicola Paoletti\inst{2} \and 
 Ezio Bartocci\inst{3} \and
 Shan Lin\inst{4} \and
 Scott A. Smolka\inst{2} \and
 Paolo Zuliani\inst{1}}
\institute{School of Computing, Newcastle University, UK \\ 
		\email{\{f.shmarov,paolo.zuliani\}@ncl.ac.uk} 
	\and Department of Computer Science, Stony Brook University, NY, USA \\
		\email{\{nicola.paoletti,sas\}@cs.stonybrook.edu}
	\and Faculty of Informatics, TU Wien, Austria\\
		\email{ezio.bartocci@tuwien.ac.at}
        	\and Department of Electrical and Computer Engineering, Stony Brook University, NY, USA \\
\email{shan.x.lin@stonybrook.edu}
}

\usepackage{pslatex}
\usepackage{graphicx}

\usepackage{array}
\usepackage{amsmath,amssymb,amscd,amstext}
\usepackage{mathrsfs}
\usepackage{dsfont}
\usepackage{amsbsy}
\usepackage{stmaryrd}
\usepackage{subfigure}

\usepackage{setspace}
\usepackage{fancyhdr} 
\usepackage{hyperref}
\usepackage{lastpage}
\usepackage{wrapfig}
\usepackage{tikz}
\usetikzlibrary{arrows,automata,positioning}

\usepackage{multirow,rotating}
\usepackage{todonotes}

\def  \Nat      {\mathbb{N}}   
\def  \rats     {\mathbb{Q}}  

\def  \ft#1     {\footnote{#1} }

\def  \eg       {{\em e.g.}}
\def  \ie       {{\em i.e.}}

\newcommand{\goes}[1]{\rightarrow{#1}}

\newcommand{\reach}{{\bf Reach}}

\renewcommand{\phi}{\varphi}

\newcommand{\reachdepth}{l}
\newcommand{\hide}[1]{}

\renewcommand{\paragraph}[1]{\vspace*{4pt}\noindent\textit{#1:}}

\DeclareMathOperator*{\argmin}{arg\,min}
\DeclareMathOperator*{\argmax}{arg\,max}

\newcommand{\review}[1]{#1}

\begin{document}

\maketitle

\begin{abstract}
We present a new method for the automated synthesis of safe and robust 
Proportional-Integral-Derivative (PID) controllers for stochastic 
hybrid systems. 
Despite their widespread use in industry, no automated method currently exists for deriving a PID controller (or any other type of controller, for that matter) with safety and performance guarantees for such a general class of systems.  In particular, we consider hybrid systems with nonlinear dynamics (Lipschitz-continuous ordinary differential equations) and random parameters, and 
we synthesize PID controllers such that 
the resulting closed-loop systems satisfy safety and performance constraints given as probabilistic bounded reachability properties.  Our technique leverages SMT solvers over the reals and nonlinear differential equations to provide formal guarantees that the synthesized 
controllers satisfy such properties. These controllers are also robust by design since they minimize the probability of reaching an unsafe state in the presence of random disturbances. 
We apply our approach to the problem of insulin regulation for type~1 diabetes, synthesizing 
controllers with robust responses to large random meal disturbances, thereby enabling them to maintain blood glucose levels within healthy, safe ranges. 

\end{abstract}


\section{Introduction}
\label{sec:intro}

Proportional-Integrative-Derivative (PID) controllers
are among the most widely deployed and well-established
feedback-control techniques.  Application areas are diverse and include industrial control systems, flight controllers, robotic manipulators, and medical devices.  
The PID controller synthesis problem entails finding the values 
of its control parameters (proportional, integral and derivative gains) 
that are optimal in terms of providing stable feedback control to the 
target system (the plant) with desired response behavior.
Despite the limited number of parameters, this problem is far from trivial,
due to the presence of multiple (and often conflicting) 
performance criteria that a controller is required to meet 
(\eg, normal transient response, stability). 

Developing PID controllers for \emph{cyber-physical systems} is even more challenging because their dynamics are typically hybrid, nonlinear, and stochastic in nature. Moreover, it is imperative that the closed-loop controller-plus-plant system is safe (\ie, 
does not reach a bad state) and robust (\ie, exhibits desired behavior under a given range of disturbances).
To the best of our knowledge, however, the current techniques for synthesizing PID controllers
(see \eg,~\cite{su2005design,duong2012robust,fliess2013model}) simply ignore these issues and do not provide any formal guarantees about the resulting closed-loop system.

In this paper, we present a new framework for the automated synthesis of PID controllers for 
{\em stochastic hybrid systems} such that the resulting closed-loop system {\em provably} satisfies 
a given (probabilistic) safety property in a robust way with respect to random disturbances. 
Specifically, we formulate and tackle two different, yet complementary, problems: \textit{controller synthesis}, 
\ie, find a PID controller that minimizes the probability of violating the property, thus ensuring robustness against random perturbations; and 
\textit{maximum disturbance synthesis}, \ie, find, for a given controller, the largest disturbance 
that the resulting control system can sustain without violating the property. 
To the best of our knowledge, we are the first to present a solution to these problems (see also the 
related work in Section~\ref{sec:related_work}) with formal guarantees. 

It is well known that safety verification is an inherently difficult 
problem for nonlinear hybrid systems \review{--- it is in general undecidable, 
hence it must be solved using approximation methods.} 
Our technique builds on the frameworks of delta-satisfiability \cite{DBLP:conf/lics/GaoAC12} and probabilistic 
delta-reachability \cite{DBLP:conf/hvc/ShmarovZ16} to reason formally about nonlinear and stochastic dynamics. This enables us to circumvent undecidability issues by returning solutions with numerical guarantees up to an arbitrary user-defined precision. 

We express safety and performance constraints as probabilistic bounded reachability properties, and encode the synthesis problems as SMT formulae over ordinary differential equations. This theory adequately captures, besides the reachability properties, the hybrid nonlinear dynamics that we need to reproduce, and leverages appropriate SMT solvers~\cite{dreal,ProbReach} that can solve the delta-satisfiability problem for such formulae. 
 
We demonstrate the utility of our approach on an artificial pancreas case study, \ie\ the closed-loop insulin regulation for type~1 diabetes. In particular, we synthesize controllers that can provide robust responses to large random meal disturbances, while keeping the blood glucose level within healthy, safe ranges. 

To summarize, in this paper, we make the following main contributions:
\begin{itemize}
	\item We provide a solution to the \emph{PID controller synthesis} and \emph{maximum disturbance synthesis} problems using an SMT-based framework that supports hybrid plants with {\em nonlinear ODEs} and {\em random parameters}. 
	\item We \review{encode} in the framework safety and performance requirements, and state the corresponding formal guarantees for the {\em automatically synthesized} PID controllers. 
	\item We demonstrate the practical utility of our approach by synthesizing provably safe and robust controllers for an artificial pancreas model. 
\end{itemize}



\section{Background}
\label{sec:problem}


Hybrid systems extend finite-state automata by introducing continuous state spaces and continuous-time dynamics \cite{DBLP:conf/hybrid/AlurCHH92}.
They are especially useful when modeling systems that combine discrete and continuous behavior such as cyber-physical systems, including
biomedical devices (\eg, infusion pumps and pacemakers). In particular, continuous dynamics is usually expressed via (solutions of)
ordinary differential equations (ODEs). To capture a wider and more realistic family of systems, in this work we consider hybrid systems 
whose behavior depends on both {\em random} and {\em nondeterministic} parameters, dubbed {\em stochastic parametric hybrid systems (SPHS)} 
\cite{DBLP:conf/hvc/ShmarovZ16}.  In particular, our synthesis approach models both the target system and its controller as a single SPHS. 
It is thus important 
to adopt a formalism that allows random {\em and} nondeterministic parameters: the former are used to model system disturbances and plant uncertainties, while the
latter are used to constrain the search space for the controller synthesis.

\begin{definition}{\bf (SPHS)}\cite{DBLP:conf/hvc/ShmarovZ16}\label{def:phs}
A Stochastic Parametric Hybrid System is a tuple $H=<Q,\Upsilon,X,P,Y,R,$ $\mathsf{jump},\mathsf{goal}>$, 
where
\begin{itemize}
\item $Q = \{q_0, \cdots , q_m\}$ is the set of modes (discrete states) of the system;
\item $\Upsilon \subseteq \{(q,q^{\prime}): q,q^{\prime} \in Q\}$ is the set of possible mode transitions (discrete dynamics);
\item $X = [u_1, v_1] \times \cdots \times [u_n, v_n]\times [0,T] \subset \mathbb{R}^{n+1}$ is the continuous system state space; 
\item $P = [a_{1}, b_{1}] \times \cdots \times [a_{k}, b_{k}] \subset \mathbb{R}^k$ is the parameter space of the system, which is
	represented as $P = P_R \times P_N$, where $P_R$ is domain of random parameters and $P_N$ is the domain of nondeterministic 
parameters (and either domain may be empty);
\item $Y = \{{{\bf y}_{q}({\bf p})} : q \in Q, {\bf p}\in X\times P\}$ is the continuous dynamics where ${\bf y}_q:X\times P \rightarrow X$; 
\item $R = \{{\bf g}_{(q,q^{\prime})}({\bf p}) : (q, q^{\prime}) \in \Upsilon, {\bf p}\in X\times P\}$ is the set of `reset' functions ${\bf g}_{(q,q^{\prime})}: X\times P \rightarrow X\times P$ defining the continuous state at time $t = 0$ in mode $q^{\prime}$ after taking the transition from mode $q$.
\end{itemize}
and predicates (or relations)
\begin{itemize}
        \item $\mathsf{jump}_{(q, q^\prime)}({\bf p})$ is true iff the discrete transition $(q, q^{\prime}) \in \Upsilon$ may occur upon reaching state $({\bf p},q)\in X\times P\times Q$,
        \item $\mathsf{goal}_q({\bf p})$ is true iff\, ${\bf p}\in X\times P$ is a goal state for mode $q$.
\end{itemize}
\end{definition}
The $\mathsf{goal}$ predicate is the same for all modes and is used to define the safety requirements for the controller synthesis 
(see (\ref{eq:prop_1}) in Section \ref{sec:implementation}). We assume that the SPHS has an initial state $({\bf x_0}, q_0) \in X\times Q$.
 The continuous dynamics $Y$ is given as an initial-value problem with Lipschitz-continuous ODEs over a bounded time domain $[0,T]$, which have a unique solution for any given initial condition ${\bf p} \in X\times P$ (by the Picard-Lindel{\"o}f theorem). System parameters are treated as variables with zero derivative, and thus are part of the initial conditions. Finally, parameters may be random discrete/continuous \review{(capturing system disturbances and uncertainties)} with an 
associated probability measure, and/or nondeterministic \review{(\ie\ the parameters to synthesize)}, in which case only their bounded domain is known. 


\paragraph{\bf Probabilistic Delta-Reachability}
\label{sec:probreach}
For our purposes we need to consider {\em probabilistic bounded} reachability: what is the {\em probability that a SPHS} (which models system {\em and} controller) {\em reaches a goal state in a {\bf finite} number of discrete transitions}? 
Reasoning about reachability in nonlinear hybrid systems entails deciding first-order formulae over the reals.
It is well known that such formulae are undecidable when they include, \eg, trigonometric functions.  
A relaxed notion of satisfiability ($\delta$-satisfiability \cite{DBLP:conf/lics/GaoAC12}) can be utilized 
to overcome this hurdle, and SMT solvers such as dReal \cite{dreal} and iSAT-ODE \cite{DBLP:conf/atva/EggersFH08} 
can ``$\delta$-decide'' a wide variety of real functions, including transcendental functions and solutions of nonlinear ODEs. 
(Essentially, those tools implement solving procedures that are sound and complete up to a given arbitrary precision.)

A probabilistic extension of bounded reachability in SPHSs was presented in \cite{DBLP:conf/hvc/ShmarovZ16}, which basically 
boils down to measuring the {\em goal set}, \ie\, the set of parameter points for which the system satisfies the reachability property. Recall that the set of goal states for a SPHS is described by its $\mathsf{goal}$ predicate. 
When nondeterministic parameters are present, 
the system may exhibit a range of reachability probabilities, depending on the value of the nondeterministic parameters. That is,
the reachability probability is given by a function ${\bf Pr}(\nu) = \int_{G(\nu)}  d\mathbb{P}$, defined for any $ \nu\in P_N $,
where $G(\nu)$ is the goal set and $\mathbb{P}$ is the probability measure of the random
parameters. The ProbReach tool utilizes the notion of $\delta$-satisfiability when computing the goal set, thereby computing 
{\em probabilistic $\delta$-reachability} \cite{ProbReach}. In particular, ProbReach computes probability {\em enclosures} for 
the range of function {\bf Pr} over parameter sets $\mathcal{N} \subseteq P_N$, \ie, intervals $[a,b]$ such that
\begin{equation}\label{def:enc}
	\forall \nu \in \mathcal{N}\quad {\bf Pr}(\nu) \in [a,b]
\end{equation}
where $0\leqslant a \leqslant b \leqslant 1$ (but $a=b$ can only be achieved in very special cases, of course). 
To solve our synthesis problems we leverage ProbReach's formal approach and statistical approach for the computation of 
probability enclosures.

\paragraph{Formal Approach} 
ProbReach guarantees that the returned enclosures satisfy (\ref{def:enc}) {\em formally} and {\em numerically} \cite{ProbReach}. 
In particular, any enclosure either has a desired width $\epsilon\in\rats^+$, or the size of the corresponding parameter box $\mathcal{N}\subseteq P_N$
is smaller than a given lower limit.
The computational complexity of this approach increases exponentially with the number of parameters, so it might not be feasible 
for large systems.

\paragraph{Statistical Approach} It trades computational complexity with correctness guarantees \cite{DBLP:conf/hvc/ShmarovZ16}, 
by solving approximately the problem of finding a value $\nu^*$ for the nondeterministic parameters that minimizes (maximizes) 
the reachability probability {\bf Pr}:
\begin{equation}\label{def:preach-opt}
	\nu^* \in \argmin_{\nu\in P_N}\, {\bf Pr}(\nu) \hspace{6ex} \big( \nu^* \in \argmax_{\nu\in P_N}\, {\bf Pr}(\nu) \big)\ .
\end{equation}
ProbReach returns an estimate $\hat{\nu}$ for $\nu^*$ and a probability enclosure $[a,b]$ that are {\em statistically} and {\em numerically} 
guaranteed to satisfy:
\begin{equation}\label{def:preach-opt1}
	\text{Prob}({\bf Pr}(\hat{\nu}) \in [a,b]) \geqslant c
\end{equation}
where $0<c<1$ is an arbitrary confidence parameter. 
In general, the size of the enclosure $[a,b]$ cannot be arbitrarily chosen due to undecidability reasons, although
it may be possible to get tighter enclosures by increasing the numerical precision of $\delta$-reachability.
Also, the statistical approach utilizes a Monte Carlo (Cross Entropy) method, so it cannot guarantee that $\hat{\nu}$ is a global 
optimum, \ie, that satisfies (\ref{def:preach-opt}).




\paragraph{\bf PID control} 
A PID control law is the sum of three kinds of control actions, \textit{Proportional, Integral and Derivative actions}, each of which depends on the \textit{error value}, $e$, \ie\ the difference between a target trajectory, or \textit{setpoint }$sp$, and the measured output of the system $y$. At time $t$, the resulting control law $u(t)$ and error $e(t)$ are given by:
\begin{equation}
u(t) = \underbrace{K_p  e(t)}_\text{P} + \underbrace{K_i \int_0^t e(\tau) \ d\tau}_\text{I} +  \underbrace{K_d \dot{e}(t)}_\text{D}, \qquad e(t) = sp(t) - y(t) \label{eq:PID}
\end{equation}
where constants $K_p$, $K_i$ and $K_d$ are called \textit{gains} and fully characterize the PID controller. 


The above control law assumes a continuous time domain, which is quite common in the design stage of a PID controller. Alternatively, PID control can be studied over discrete time, where the integral term is replaced by a sum and the derivative by a finite difference. 
\review{However, the analysis of discrete-time PID controllers is impractical for non-trivial time bounds because they induce a discrete transition for each time step, and thus, they directly affect the unrolling/reachability depth required for the bounded reachability analysis, which is at the core of our synthesis method.}



\section{PID Control of Hybrid Plants}
\begin{wrapfigure}{r}{0.6\textwidth}
\centering
\vspace*{-1.5em}
\includegraphics[width=.6\textwidth]{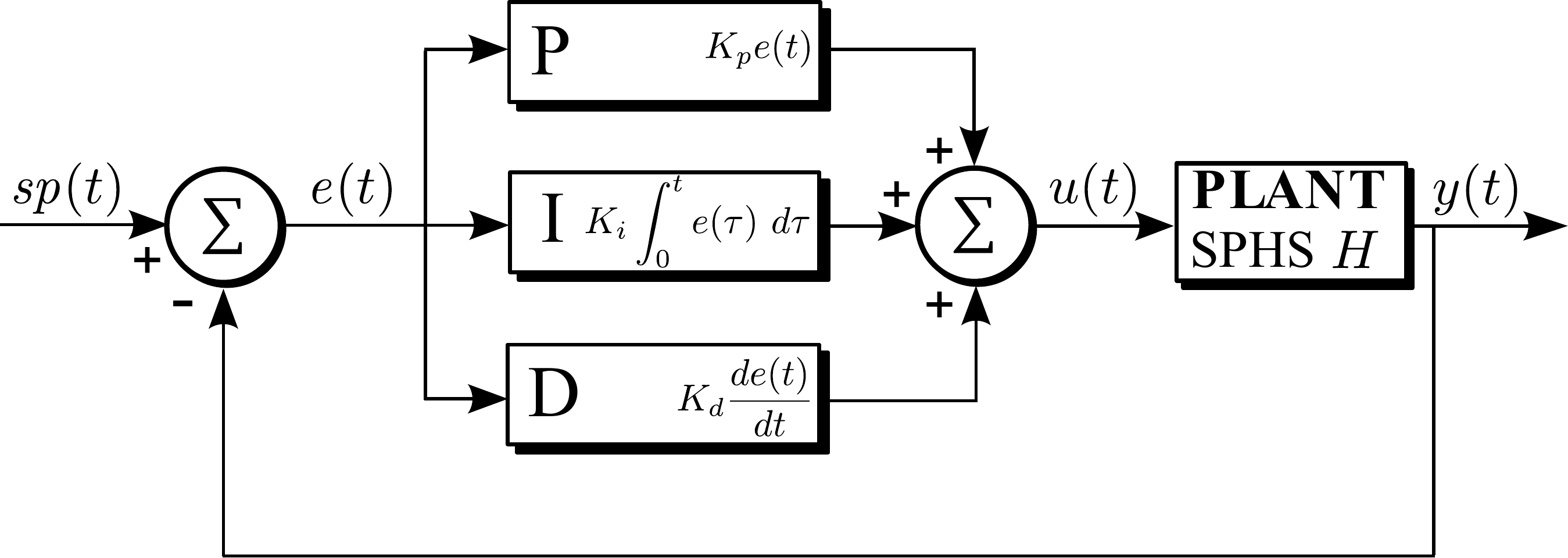}
\caption{PID control loop}\label{fig:PID}
\vspace*{-1.5em}
\end{wrapfigure}
We formally characterize the system  given by the feedback loop between a plant SPHS $H$ and a PID controller, so called {\em closed-loop system} (see Figure \ref{fig:PID}).
 \review{We would like to stress that we support plants specified as hybrid systems, given that a variety of systems naturally exhibit hybrid dynamics (regardless of the controller). For instance, in the artificial pancreas case study of Section~\ref{sec:evaluation}, discrete modes are used to describe different meals, while the glucose metabolism is captured by a set of ODEs.}

We assume that the controller is an additive input and can manipulate only one of the state variables of $H$, $x_u$, and that for each mode $q$ of $H$, there is a measurement function $h_q$ that provides the output of the system at $q$. To enable synthesis, we further assume that the PID controller gains $\mathbf{k} = (K_p,K_i,K_d)$ are (unknown) nondeterministic parameters with domain $K$. To stress this dependency, below we use the notation $u(\mathbf{k},t)$ to denote the PID control law of Equation~\ref{eq:PID}.

\begin{definition}[PID-SPHS control system] Let $H= \langle Q$, $\Upsilon$, $X$, $P$, $Y$, $R$, $\mathsf{jump},\mathsf{goal}\rangle$ be a plant SPHS, and let $u$ be a PID controller (\ref{eq:PID}) with gain parameters $\mathbf{k} \in K \subset \mathbb{R}^3$. For $q \in Q$, let $h_q: X \goes{} \mathbb{R}$ be the corresponding measurement function. Let $x_u$ be the manipulated state variable, $i_u \in \{1,\ldots,n\}$ be the corresponding index in the state vector, and $sp: [0,t] \goes{} \mathbb{R}$ be the desired setpoint. The {\em PID-SPHS control system} with plant $H$ is the SPHS $H \parallel u = \langle Q, \Upsilon, X, P\times K, Y', R', \mathsf{jump},\mathsf{goal}\rangle$, where
\begin{itemize}
\item $Y' = \{ \mathbf{y}'_q(\mathbf{p},\mathbf{k},t) : q \in Q,  \mathbf{p} \in X \times P, \mathbf{k} \in K, t \in [0,1]\}$, where the continuous dynamics of each state variable with index $i=1,\ldots, n$ is given by
\[
\mathbf{y}'_{q,i}(\mathbf{p},\mathbf{k},t) = 
\begin{cases}
\mathbf{y}_{q,i}(\mathbf{p},t) + u(\mathbf{k},t) & \text{if $i = i_u$}\\
\mathbf{y}_{q,i}(\mathbf{p},t) & \text{otherwise}
\end{cases}
\]
where $\mathbf{y}_{q,i}$ is the corresponding continuous dynamics in the plant SPHS $H$, and $u(\mathbf{k},t)$ is the PID law described in  (\ref{eq:PID}), with error $$e(t) = sp(t) - h_q(\mathbf{y}'_{q}(\mathbf{p},\mathbf{k},t)); \text{ and}$$
\item $R' = \{ \mathbf{g}'_{(q,q')}(\mathbf{p},\mathbf{k},t) : (q,q') \in \Upsilon,  \mathbf{p} \in X \times P, \mathbf{k} \in K, t \in [0,T]\}$, where $\mathbf{g}'_{(q,q')}(\mathbf{p},\mathbf{k},t) = \mathbf{g}_{(q,q')}(\mathbf{p},t)$, i.e.\ the reset $\mathbf{g}'_{(q,q')}$ is not affected by the controller parameters $\mathbf{k}$ and is equal to the corresponding reset of the plant $H$, $\mathbf{g}_{(q,q')}$.
\end{itemize}
\end{definition}

\review{In other words, the PID-SPHS control system is obtained by applying the same PID controller to the continuous dynamics of each discrete mode of the hybrid plant, meaning that the PID-SPHS composition produces the same number of modes of the plant SPHS. } 
We remark that external disturbances as well as plant uncertainties can be encoded through appropriate random variables in the plant SPHS.


\section{Safe and Robust PID Controller Synthesis}
\label{sec:implementation}
In this section we first illustrate the class of synthesis properties of interest, able to capture relevant safety and performance objectives. Second, we formulate the PID control synthesis problem and the related problem of {maximum disturbance synthesis}. 

We remark that our main objective is designing PID controllers with formal \textbf{safety guarantees}, \ie\ a given set of bad states should never be reached by the system, or reached with very small probability. 
Similarly, we aim to synthesize controllers able to guarantee, by design, prescribed performance levels. For instance, the designer might need to keep the settling time within strict bounds, or avoid large overshoot. 

To this purpose, we consider two well-established performance measures, the fundamental index ($FI$) and the weighted fundamental index ($FI_w$)~\cite{levine1996control,li2004cautocsd}\footnote{$FI$ and $FI_w$ are also also known as ``integral of square error'' and ``integral of square time weighted square error'', respectively.}, defined by:
\begin{equation}\label{eq:FI}
FI(t) = \int_0^t \left(e(\tau)\right)^2 d\tau \qquad FI_w(t) = \int_0^t \tau^2 \cdot \left(e(\tau)\right)^2 d\tau.
\end{equation}
$FI$ and $FI_w$ quantify the cumulative error between output and set-point, thus providing a measure of 
how much the system deviates from the desired behavior. Crucially, they also indirectly capture key transient response measures such as {steady-state error}, \ie\ the value of $e(t)$ when $t\goes{}\infty$, 
or maximum overshoot, \ie\ the highest deviation from the setpoint\footnote{In PID theory, transient response measures are often evaluated after applying a step function to the set-point. However, we do not restrict ourselves to this scenario.}. In fact, small $FI$ values typically indicate good transient response (\eg\ small overshoot or short rise-time), while $FI_w$ weighs errors with the corresponding time, in this way stressing steady state errors.

We now formulate the main reachability property for the synthesis of safe and robust controllers, which is expressed by predicate $\mathsf{goal}$. The property captures the set of bad states that the controller should avoid (predicate $\mathsf{bad}$) as well as performance constraints through upper bounds $FI^{\max}, FI_w^{\max} \in \mathbb{R}^+$ on the allowed values of $FI$ and $FI_w$, respectively, and is given by: 
\begin{equation}\label{eq:prop_1}
\mathsf{goal} = \mathsf{bad} \vee (FI > FI^{\max}) \vee ( FI_w > FI_w^{\max}).
\end{equation}
In the case that the designer is not interested in constraining $FI$ or $FI_w$, we allow $FI^{\max}$ and $FI_w^{\max}$ to be set $+\infty$.

We now introduce the PID controller synthesis problem that aims at synthesizing the control parameters yielding the minimal probability of reaching the goal (\ie\ the undesired states). Importantly, this corresponds to minimizing the effects on the plant of random disturbances, that is, to \textit{maximizing the robustness} of the resulting system. 

We remark that the unrolling depth and the $\mathsf{goal}$ predicate are implicit in the reachability probability function $\bf Pr$ (see Section~\ref{sec:problem}). 

\begin{problem}[PID controller synthesis]\label{prob:controller_synth}
Given a PID-SPHS control system $H \parallel u$ with unknown control parameters $\mathbf{k} \in K$, find the parameters $\mathbf{k}^*$ that minimize the probability of reaching the goal:
\[
\mathbf{k}^* \in \argmin_{\mathbf{k} \in K} \ {\bf Pr}(\mathbf{k}).
\]
\end{problem}
For the duality between safety and reachability, Problem~\ref{prob:controller_synth} is equivalent to synthesizing controllers that maximize the probability that $\neg \mathrm{goal}$ always holds.
If $H \parallel u$ has no random parameters (but only nondeterministic parameters), then Problem~\ref{prob:controller_synth} is equivalent to synthesizing, if it exists, a controller that makes goal unsatisfiable. 
\hide{
\begin{remark}
For the duality between safety and reachability, Problem~\ref{prob:controller_synth} is equivalent to synthesizing controllers that maximize the probability that $\neg \mathrm{goal}$ always holds.
\end{remark}

\begin{remark}
If $H \parallel u$ has no random parameters (but only nondeterministic parameters), then Problem~\ref{prob:controller_synth} is equivalent to synthesizing, if it exists, a controller that makes goal unsatisfiable. 
\end{remark}
}

As previously explained, the control parameters $\mathbf{k}$ that we aim to synthesize must be defined as nondeterministic parameters in the SPHS $H \parallel u$. Crucially, we can employ both the formal and the statistical approach alike to solve this problem. 

In general, it is not possible to know the exact minimizing parameter because of the inherent undecidability. 
However, using the formal approach one could select the synthesized controller parameter $\mathbf{k}^*$ \textit{as the midpoint of the parameter box whose enclosure has the least midpoint}. Through the following proposition, we show that this solution can be made arbitrarily precise when all of the returned enclosures have length $\leq \epsilon$, the user-defined 
parameter that determines the desired length of the enclosure as explained in Section \ref{sec:probreach} (however, this cannot be always guaranteed).
\begin{proposition}\label{prop:formal}
Suppose that the returned enclosures by the formal approach have all length $\leq \epsilon$.
Let $P^*$ be the actual minimal probability, and let $\mathbf{k}^*$ be the solution of the formal approach for Problem \ref{prob:controller_synth}. Then, it holds that
$$\mathbf{Pr}(\mathbf{k}^*) < P^* + \frac{3}{2} \epsilon\  .$$  
\end{proposition}
\begin{proof}
See Appendix \ref{app:proof1}.
\end{proof}

On the other hand, the statistical algorithm returns an over-approximation $\hat{P}$ of the minimum probability, $c$-confidence interval $[\hat{P}]$ such that $\hat{P} \in [\hat{P}]$, and synthesized parameters $\mathbf{k}^*$ whose reachability probability is included in $[\hat{P}]$ with probability at least $c$, as per Equations \ref{def:preach-opt} and \ref{def:preach-opt1}.



Below, we define the maximum disturbance synthesis problem, aimed at finding, given a concrete controller, the maximum disturbance value that the resulting control system can support without violating a given property. This problem is complementary to the PID synthesis problem, since it allows the designer to formally evaluate the robustness of a known controller, possibly synthesized in a previous step. 
Specifically, we assume that the disturbance is represented by a vector of nondeterministic parameters $\mathbf{d}$ in the plant SPHS, and that $\mathbf{d}$ ranges over some bounded domain~$D$. 

\begin{problem}[Maximum disturbance synthesis]\label{prob:dist_synth}
Given a PID-SPHS control system $H \parallel u$ with {\em known} control parameters $\mathbf{k}^* \in K$ and {\em unknown} disturbance $\mathbf{d} \in D$, and a probability threshold $p$, find the highest disturbance $\mathbf{d}^*$ for which the probability of reaching the goal does not exceed $p$, \ie\ such that:
\[
\mathbf{d}^* = \max \left\lbrace \mathbf{d} \in D \mid \mathbf{Pr}(\mathbf{d}) \leq p \right\rbrace.
\]
\end{problem}
For the duality between safety and reachability, the probability of reaching $\mathrm{goal}$ is below $p$ if and only if the probability that $\neg\mathrm{goal}$ always holds is above $1-p$.
If $H \parallel u$ has no random parameters (but only nondeterministic parameters), then Problem~\ref{prob:dist_synth} reduces
to finding the largest disturbance for which the PID-SPHS system either reaches or does not reach the goal. 

\hide{
\begin{remark}
For the duality between safety and reachability, the probability of reaching $\mathrm{goal}$ is below $p$ if and only if the probability that $\neg\mathrm{goal}$ always holds is above $1-p$.
\end{remark}

\begin{remark}
If $H \parallel u$ has no random parameters (but only nondeterministic parameters), then Problem~\ref{prob:dist_synth} reduces
to finding the largest disturbance for which the PID-SPHS system either reaches or does not reach the goal. 
\end{remark}
}

Note that the maximum disturbance synthesis problem is fundamentally different from the controller synthesis problem, because the kind of parameters that we seek to synthesize represent external factors that cannot be controlled. That is why we are interested in knowing the maximum (worst-case) value they can attain such that the requirements are met with given probability constraints. In particular, we restrict to upper-bound constraints because we want to limit the probability of reaching a given goal (undesired) state, even though lower bound constraints can be equally supported by the synthesis method.

Problem~\ref{prob:dist_synth} is solved through the formal approach, which allows identifying the parameters boxes whose probability enclosures are guaranteed to be below the threshold $p$, \ie, they are intervals of the form $[P_{\min},P_{\max}]$ with $P_{\max}\leq p$. Then, the synthesized parameter $\mathbf{d}^*$ is selected as the highest value among all such parameter boxes. 

It follows that the returned $\mathbf{d}^*$ is guaranteed to meet the probability constraint ($\mathbf{Pr}(\mathbf{d}^*) \leq p$), but, due to the iterative refinement, $\mathbf{d}^*$ under-estimates the actual maximum disturbance. In this sense, $\mathbf{d}^*$ is a safe under-approximation. The reason is that there might exist some ``spurious'' parameter boxes $[\mathbf{d}]$ (not returned by the algorithm), \ie\ such that $p$ lies within the corresponding probability enclosure $[P]$ and $[\mathbf{d}]$ contains a disturbance value $\mathbf{d}'$ that is higher than the synthesized $\mathbf{d}^*$ and that, at the same time, meets the constraint $\mathbf{Pr}(\mathbf{d}') \leq p$.

The statistical approach cannot be applied in this case, because it relies on the Cross Entropy method, which is designed for estimation and optimization purposes and is not suitable for decision problems. Note indeed that the probability bound $\leq p$ induces a Boolean (and not quantitative) property.


\section{Case Study: Artificial Pancreas}
\label{sec:evaluation}

We evaluate our method on the closed-loop control of insulin treatment for Type 1 diabetes (T1D), also known as the \textit{artificial pancreas (AP)}~\cite{hovorka2011closed}. Together with model predictive control, PID is the main control technique for the AP~\cite{steil2011effect,huyett2015design}, and is found as well in commercial devices~\cite{kanderian2014apparatus}. 

The main requirement for the AP is to keep blood glucose (BG) levels within tight, healthy ranges, typically between 70-180 mg/dL, in order to avoid \textit{hyperglycemia} (BG above the healthy range) and \textit{hypoglycemia} (BG below the healthy range). While some temporary, postprandial hyperglycemia is typically admissible, hypoglycemia leads to severe health consequences, and thus, it should be avoided as much as possible. This is a crucial safety requirement, which we will incorporate in our synthesis properties.

The AP consists of a continuous glucose monitor that provides glucose measurements to a control algorithm regulating the amount of insulin injected by the insulin pump. 
The pump administers both \textit{basal insulin}, a low and continuous dose that covers insulin needs outside meals, and \textit{bolus insulin}, a single high dose for covering meals. 

Meals represent indeed the major disturbance in insulin control, which is why state-of-the-art commercial systems\footnote{MINIMED 670G by Medtronic \url{https://www.medtronicdiabetes.com/products/minimed-670g-insulin-pump-system}} can only regulate basal insulin and still require explicit meal announcements by the patient for bolus insulin. To this purpose, robust control methods have been investigated \cite{parker2000robust,szalay2014linear,paolettiCMSB17}, since they are able to minimize the impact of input disturbances (in our case, meals) on the plant (the patient). Thus, they have the potential to provide full closed-loop control of bolus insulin without manual dosing by the patient, which is inherently error-prone and hence, dangerous. Our method for the synthesis of safe and robust controllers is therefore particularly meaningful in this case. 

\subsection{Plant Model} 
To model the continuous system's dynamics (\eg, glucose and insulin concentrations), we consider the well-established nonlinear model of Hovorka {\em et al}.~\cite{Hovorka04}.

At time $t$, the input to the system is the infusion rate of bolus insulin, $u(t)$, which is computed by the PID controller.  The system output $y(t)$ is given by state variable $Q_1(t)$ (mmol), describing the amount of BG in the accessible compartment, \ie\ where measurements are taken, for instance using finger-stick blood samples. For simplicity, we did not include a model of the continuous glucose monitor (see \eg~\cite{wilinska2010simulation}) that instead measures glucose in the tissue fluid, but we assume continuous access to blood sugar values. 
The state-space representation of the system is as follows:
\begin{equation}
\dot{\mathbf{x}}(t) = {\bf F}\left(\mathbf{x}(t), u(t), D_G \right), \qquad y(t) = Q_1(t) \label{eq:state_space_form}
\end{equation}
where $\mathbf{x}$ is the 8-dimensional state vector that evolves according to the nonlinear ODE system $\bf{F}$ (see Appendix~\ref{app:model} for the full set of equations and parameters). The model assumes a single meal starting at time $0$ and consisting of an amount $D_G$ of ingested carbohydrates. Therefore, parameter $D_G$ represents our input disturbance. 

Instead of the BG mass $Q_1(t)$, in the discussion of the results we will mainly evaluate the BG concentration $G(t) = Q_1(t)/V_G$, where $V_G$ is the BG distribution volume. 

The error function of the PID controller is defined as $e(t) = sp - Q_1(t)$ with the constant set point $sp$ corresponding to a BG concentration of $110$ mg/dL.  
Multiple meals can be modeled through a stochastic parametric hybrid system with one mode for each meal. In particular, we consider a one-day scenario consisting of three random meals (breakfast, lunch and dinner), resulting in the SPHS of Figure~\ref{fig:model-phs}.

\begin{figure}
\centering
\scalebox{0.9}{
\begin{tikzpicture}
    \node[state, draw=none] at (-3, 0)  (init)    {};
    \node[state] at (0, 0)  (m1) {\footnotesize Meal 1};
	\node[state] at (4, 0)  (m2) {\footnotesize Meal 2};
    \node[state] at (8, 0)  (m3) {\footnotesize Meal 3};

    \draw[every loop,
              auto=right,
              >=latex]
    (init) edge[pos=0.5, sloped, below] node {\scriptsize $D_G:=D_{G_1}$} (m1)          
    (m1) edge node[pos=0.5, above, sloped] {\scriptsize $t=T_1$} node[pos=0.5, below, sloped] {\scriptsize $(D_G:=D_{G_2})\wedge(t:=0)$} (m2)
    (m2) edge node[pos=0.5, above, sloped] {\scriptsize $t=T_2$} node[pos=0.5, below, sloped] {\scriptsize $(D_G:=D_{G_3})\wedge(t:=0)$} (m3);
\end{tikzpicture}
}
\caption{Stochastic parametric hybrid system modelling a scenario of 3 meals over 24 hours. Above each edge, we report the corresponding jump conditions, below, the resets.}
\label{fig:model-phs}
\end{figure}
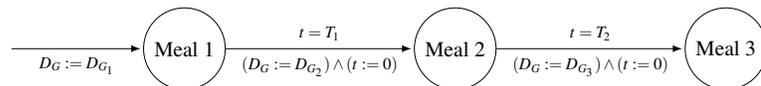

The model features five random, normally-distributed parameters: the amount of carbohydrates of each meal,
$D_{G_1}\sim\mathcal{N}(40,10)$, $D_{G_2}\sim\mathcal{N}(90,10)$ and $D_{G_3}\sim\mathcal{N}(60,10)$, 
and the waiting times between meals, $T_1 \sim \mathcal{N}(300,10)$ 
and $T_2 \sim \mathcal{N}(300,10)$.

A meal containing $D_{G_1}$ grams of carbohydrates is consumed at time 0. When the time in the first mode
reaches $T_1$ minutes the system makes a transition to the next mode $Meal~2$ where the value of the 
variable $D_G$ is set to $D_{G_2}$ and the time is reset to 0. Similarly, the system transitions from
mode $Meal~2$ to $Meal~3$, resetting variables $D_G$ and $t$ to $D_{G_3}$ and 0, respectively. All remaining variables are not reset at discrete transitions. 



\paragraph{Basal insulin and initial state} The total insulin infusion rate is given by $u(t) + u_b$ where $u(t)$ is the dose computed by the PID controller, and $u_b$ is the basal insulin. As typically done, the value of $u_b$ is chosen in order to guarantee a steady-state BG value of $Q_1 = sp$, and the steady state thus obtained is used as the initial state of the system. 

We denote with $C_0$ the basal controller that switches off the PID controller and applies only $u_b$
(\ie, $K_p$, $K_i$ and $K_d$ are equal to 0). 

\subsection{Experiments}
We apply the formal and statistical techniques of ProbReach to synthesize 
the controller parameters $K_p$, $K_d$ and $K_i$
(Problem~\ref{prob:controller_synth}) and the maximum safe disturbance $D_G$ (Problem~\ref{prob:dist_synth}), 
considering the probabilistic reachability property of Section \ref{sec:implementation}. 
All experiments in this section were conducted on a 32-core (Intel Xeon 2.90GHz) Ubuntu 16.04 machine,
and the obtained results for the synthesized controllers are summarized in Table~\ref{table:result-pid}. We also validate and assess performance of the controllers over multiple random instantiations of the meals, which is reported in Figure~\ref{fig:perf_eval}.

\subsubsection{PID controller synthesis}
Typical healthy glucose levels vary between 4 and 10 mmol/L. Since avoiding hypoglycemia ($G(t) < 4$ mmol/L) is the main safety requirement of the artificial pancreas, while (temporary) hyperglycemia can be tolerated and is inescapable after meals, we will consider a BG range of $[4, 16]$ for our safety properties. In this way we protect against both hypoglycemia and very severe levels of hyperglycemia.

Given that the basal insulin level is insufficient to cover meal disturbances, the basal controller $C_0$ prevents hypoglycemia but causes severe hyperglycemia when a large meal is consumed ($D_G > 80$) or when the BG level is not low enough by the time the next meal is consumed (see Figure \ref{fig:perf_eval}). 

We used the statistical engine of ProbReach to synthesize several controllers (see Table \ref{table:result-pid}), over domains $K_d \in [-10^{-1},0]$, 
$K_i \in [-10^{-5},0]$ and $K_p \in [-10^{-3},0]$, which minimize the probability of reaching a bad state at any time instant in the modes $Meal~1$, $Meal~2$ and $Meal~3$ (reachability depth of 0, 1 or 2, respectively). 


The set of unsafe glucose ranges is captured by predicate $\mathsf{bad} = G(t)\not\in[4,16]$. Controller $C_1$ was synthesized considering only safety requirements, corresponding to the reachability specification $\mathsf{goal} = \mathsf{bad}$ (see Equation~\ref{eq:prop_1}). On the other hand, controllers $C_2$, $C_3$ and $C_4$ were obtained taking into account also performance constraints, by using the default specification (\ref{eq:prop_1}): 
$\mathsf{goal} = \mathsf{bad} \vee (FI > FI^{\max}) \vee ( FI_w > FI_w^{\max})$. Thresholds $FI^{\max}$ and $FI_w^{\max}$ have been set to gradually stricter values, respectively to $3.5\times 10^6$ and $70\times 10^9$ for $C_2$, $3\times 10^6$ and $50\times 10^9$ for $C_3$, and $2.7\times 10^6$ and $30\times 10^9$ for $C_4$. 


\begin{table}
\footnotesize
\centering

\begin{tabular}{|c||c|c|c|c||c|c||c|c|}
\hline
\# & $K_d$ {\scriptsize $(\times 10^{2})$} & $K_i$ {\scriptsize $(\times 10^{7})$} & $K_p$ {\scriptsize $(\times 10^{4})$} & $CPU_{syn}$ & $P$ & $CPU_P$ & $D_{G_1}^{max}$& $CPU_{max}$\\
\hline
\hline
$C_0$ &0 & 0 & 0 & 0 & [0.97322,1] & 176 & 69.4 & 2,327 \\
$C_1$ & -6.02 & -3.53 & -6.17 & 92,999 & [0.19645,0.24645] & 4,937 & 88.07 & 3,682 \\
$C_2$ & -5.73 & -3.00 & -6.39 & 156,635 & [0.31307,0.36307] & 64,254 & 87.62 & 3,664 \\
$C_3$ & -6.002 & -1.17 & -6.76 & 98,647 & [0.65141,0.70141] & 59,215 & 88.23 & 3,881 \\
$C_4$ & -6.24 & -7.55 & -5.42 & 123,726 & [0.97149,1] & 11,336 & 88.24 & 3,867 \\
\hline
\end{tabular}
\vspace{1ex}
\caption{Results of PID controller synthesis, where: $\#$ -- name of the synthesized controller, 
$K_d$, $K_i$ and $K_p$ -- synthesized values of the gain constants characterizing the corresponding controller (Problem 1),
$CPU_{syn}$ -- CPU time in seconds for synthesizing the controller parameters,
$P$ -- 99\%-confidence interval for the reachability probability,
$CPU_P$ -- CPU time in seconds for computing $P$ for synthesized controller, 
$D_{G_1}^{max}$ -- synthesized maximum meal disturbance for which the system never reaches the unsafe state, 
$CPU_{max}$ -- CPU time in seconds for obtaining $D_{G_1}^{max}$.}
\label{table:result-pid}
\end{table}

Due to the high computational complexity of the artificial pancreas model, 
the controller synthesis was performed in two steps. First, the values of $K_p$, $K_i$ and $K_d$ 
were synthesized using a coarse precision (\ie, desired width for confidence intervals $P$) for computing the probability 
estimates during the nondeterministic parameter search. Second, the confidence intervals for the obtained controllers
were computed with a higher precision. The values of $CPU_{syn}$ and $CPU_{P}$ in Table \ref{table:result-pid}
represent CPU times used for solving these two steps. 
The high computation times are due to the fact
that the solvers incorporated by ProbReach solve ODEs in a guaranteed manner which
is, for general Lipschitz-continuous ODEs, a PSPACE-complete problem, and thus,
it is the main bottleneck of the implemented algorithms.

Besides $C_0$ that unsurprisingly yields the highest probability of safety violation (highest $P$ for the reachability probability), results in Table~\ref{table:result-pid} evidence 
that controllers $C_1, \ldots, C_4$ fail to maintain the safe state with increasingly higher probability. As we shall see in more detail later, this behaviour is mostly due to the performance constraints that become harder and harder to satisfy. 


\subsubsection{Maximum disturbance synthesis}

We solve Problem~\ref{prob:dist_synth} for each of the obtained controllers in Table \ref{table:result-pid}. 
We consider a domain of $[0,120]$ for the maximum meal disturbance, and apply the formal 
approach of ProbReach for synthesizing the maximum size $D_{G_1}^{max}$ of the first meal, such that, 
given any disturbance $D_{G_1} \in [0,D_{G_1}^{max}]$, the system does not reach the unsafe state 
within 12 hours. Note that this corresponds to setting the probability threshold $p$ of Problem~\ref{prob:dist_synth} to 0. Since we are interested in just one meal, we consider a reachability depth of 0 (path length of 1) for the bounded reachability property.

The results in Table \ref{table:result-pid} indicate that applying a PID controller increases
the size of the allowed meal from approximately 69g of the basal controller to about 88g, and at the same time, the difference between the synthesized controllers is negligibly small. 

Although introducing a controller does not increase the maximum disturbance dramatically 
with respect to the basal case, a PID control decreases the BG level sufficiently enough 
so that a subsequent meal of similar size 
can be consumed without the risk of experiencing severe hyperglycemia. In contrast,
$C_0$ does not bring the glucose level low enough before the following meal.

Note that, being normally distributed with mean 90 g, the second random meal exceeds such obtained maximum disturbances, which explains why the synthesized controllers fail with some probability to avoid unsafe states. 



\subsubsection{Performance and safety evaluation} \label{sec:eval-perf}
In this experiment, we evaluate safety and performance of the controllers by simulating 1,000 instantiations of the random meals. Such obtained glucose profiles and statistics are reported in Figure~\ref{fig:perf_eval}. No hypoglycemia episode ($G < 4$) was registered.

Plots evidence that all four synthesized controllers ($C_1,\ldots, C_4$) perform dramatically better than the basal controller $C_0$, which stays, on the average, $23.59\%$ of the time in severe hyperglycemia (see index $t_{\mathsf{bad}}$). In particular, all the traces simulated for $C_0$ violate the safe BG constraints $G \in [4,16]$ (100\% value of $\%_{\mathsf{bad}}$).  

On the other hand, controllers $C_1,\ldots, C_4$ violate safe BG constraints for 17-22\% of their traces, but this happens only for a very short while (no more than 0.45\% of the time) after the second (the largest) meal. This comes with no surprise since we already formally proven that the second meal exceeds the allowed maximum meal disturbance. 

$C_0$ has the worst performance in terms of $FI$ and $FI_w$, with mean $FI$ and $FI_w$ values (indices $\overline{FI}$ and $\overline{FI_w}$, resp.) significantly larger than those of $C_1,\ldots, C_4$. Among the synthesized controllers, $C_3$ has the best steady-state behavior (as visible in Figure~\ref{fig:perf_eval}, plot d), keeping the glucose level very close to the set point towards the end of the simulation. $C_3$ yields indeed the best mean $FI_w$ value (index $\overline{FI_w}$), while the worse steady-state behavior is observed for $C_4$. On the other hand, mean $FI$ values are very similar, meaning that $C_1,\ldots, C_4$ maintain the BG levels equally 
far from the set point on the average. 

One would expect $C_4$ to have the best performance in terms of $FI_w$, since it was synthesized with the stricter $FI_w$ constraint ($FI_w^{\max} = 30\times 10^9$). This constraint is, however, too strong to be satisfied, as demonstrated by the 100\% value of index $\%_{FI_w > FI_w^{\max}}$ (see Figure~\ref{fig:perf_eval}), implying all traces fail to satisfy $FI_w \leq FI_w^{\max}$. In general, we observe that strengthening the performance constraints leads to higher chances of violating them (see the last three indices of Figure~\ref{fig:perf_eval}). We conclude that performance constraints (and their violation) largely contribute to the reachability probabilities computed by ProbReach (see Table \ref{table:result-pid}) for $C_2, C_3$ and $C_4$, whose traces violate $FI$ or $FI_w$ constraints for 28\%, 67\%, and 100\% of the times, respectively.

\begin{figure}
\centering
\subfigure[$C_0$]{\includegraphics[width=.19\textwidth]{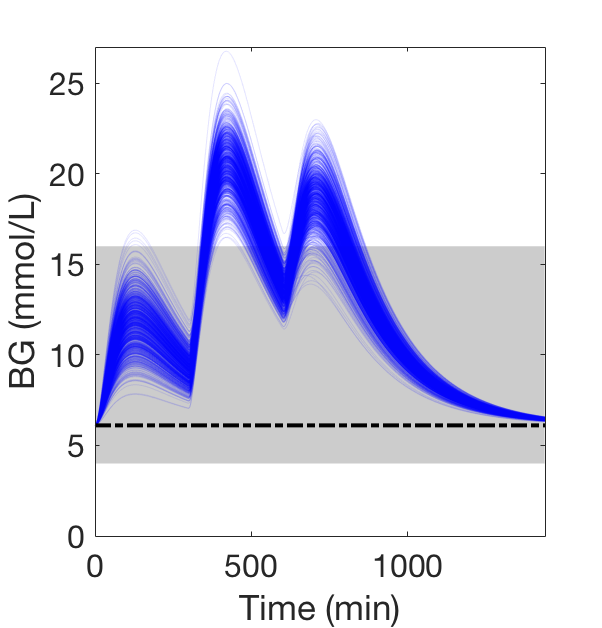}}
\subfigure[$C_1$]{\includegraphics[width=.19\textwidth]{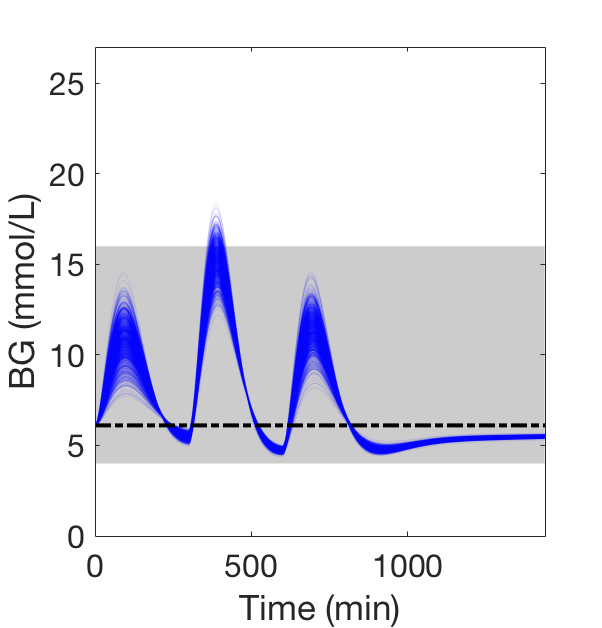}}
\subfigure[$C_2$]{\includegraphics[width=.19\textwidth]{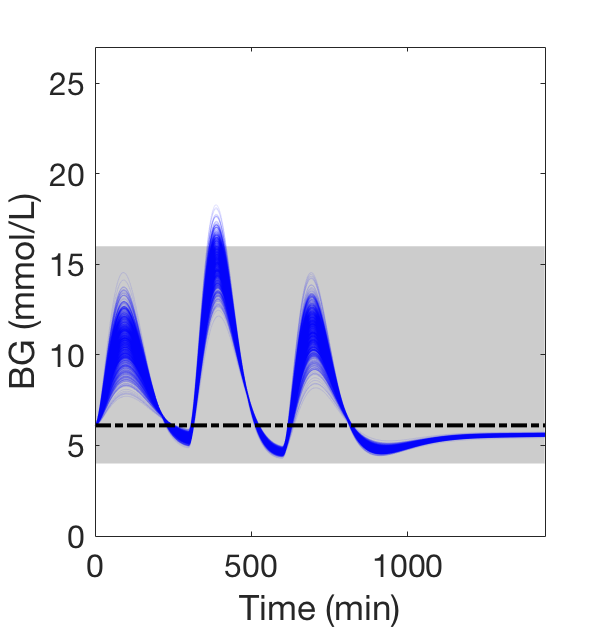}}
\subfigure[$C_3$]{\includegraphics[width=.19\textwidth]{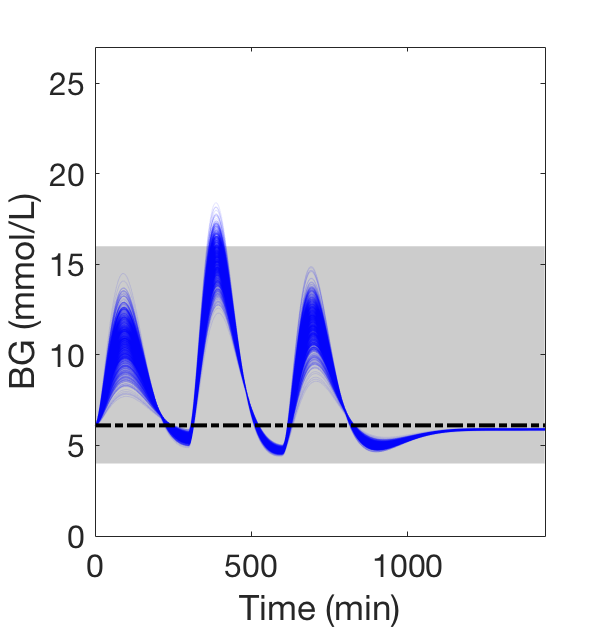}}
\subfigure[$C_4$]{\includegraphics[width=.19\textwidth]{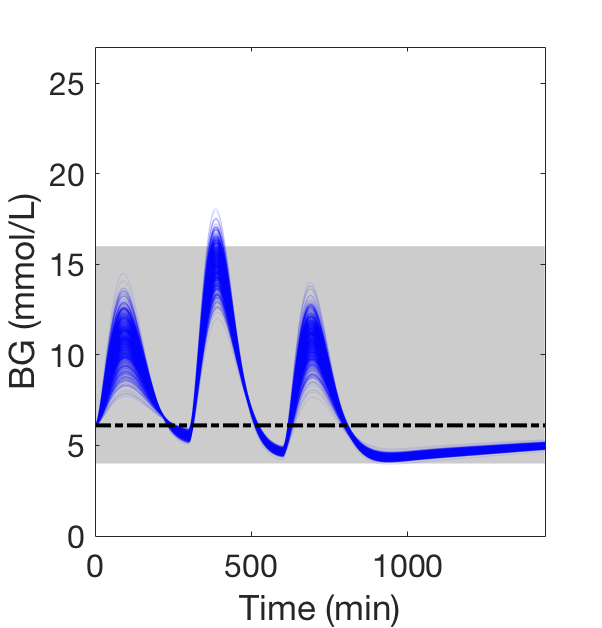}}\\
\footnotesize
\begin{tabular}{r|ccccccc}
& $t_{\mathsf{bad}}$ & 
$\%_{\mathsf{bad}}$ &
$\overline{FI}$ {\scriptsize($\times 10^{-6}$)} & 
$\overline{FI_w}$  {\scriptsize($\times 10^{-9}$)} & 
$\%_{FI > FI^{\max}}$ & $\%_{FI_w > FI_w^{\max}}$ & 
$\%_{FI > FI^{\max} \vee FI_w > FI_w^{\max}}$ \\ \hline
$C_0$ & 23.59\% & 100\% & 20.27 & 653.89 & NA & NA & NA\\
$C_1$ & 0.45\% & 22\% & {\bf 3.21} & 66.32 & NA & NA & NA\\
$C_2$ & 0.45\% & 21.4\% & 3.21 & 60.91 & {\bf 28.5\%} & {\bf 14\%}& \textbf{28.5\%} \\
$C_3$ & 0.51\% & 24.2\% & 3.24 & {\bf 44.93} & 67.2\% & 21.7\% & 67.2\% \\
$C_4$ & {\bf 0.35\%} & {\bf 17.3\%} & 3.21 & 129.05 & 86.5\% & 100\% & 100\% \\ \hline
\end{tabular}
\caption{BG profiles simulated for 1,000 random meals (shaded blue lines). Grey areas indicate healthy BG ranges ($G \in [4,16]$). Dashed black lines indicate the ideal setpoint. $t_{\mathsf{bad}}$: mean proportion of time where $G \not\in [4,16]$ (all traces yielded $G > 4$, \ie\ no hypoglycemia). $\%_{\mathsf{bad}}$: proportion of traces violating $G \in [4,16]$. $\overline{FI}$ and $\overline{FI_w}$: mean $FI$ and $FI_w$, resp. $\%_{FI > FI^{\max}}$, $\%_{FI_w > FI_w^{\max}}$ and $\%_{FI > FI^{\max} \vee FI_w > FI_w^{\max}}$: proportion of traces violating, resp., either and both performance constraints. The best value for each index is highlighted in bold.}
\label{fig:perf_eval}
\end{figure}





\section{Related Work}
\label{sec:related_work}
A number of approaches have been proposed for the PID control of nonlinear and stochastic systems. Among these, nonlinear PID control~\cite{su2005design} defines the controller gains as nonlinear functions of the system state, even though performance guarantees have been established only for subclasses of nonlinear systems. 
Adaptive PID (APID) control~\cite{fliess2013model} supports nonlinear plants with partly unknown dynamics, but no requirements can be guaranteed by design since the unknown dynamics is estimated via sampling the plant output. In contrast, we can synthesize controllers with guaranteed performance for a large class of nonlinear systems (Lipschitz-continuous) while retaining the complete system dynamics. This allows for a fully model-based approach to controller synthesis, which is key in safety-critical applications, where, on the contrary, the model-free online tuning of APID is potentially dangerous. 


PID control for Markov jump systems, \ie\ where the plant is a linear system with stochastic coefficients, is solved as a convex optimization problem in~\cite{guo2005pid,he2011robust}, while in~\cite{duong2012robust}, robust PID control for stochastic systems is reduced to a constrained nonlinear optimization problem. 
Compared to these approaches, we support models where stochasticity is restricted to random (both discrete and continuous) parameters, with nondeterministic (\ie,
arbitrary) parameters and much richer nonlinear dynamics.
Another key strength of our method with respect to the above techniques is that design specifications are given in terms of probabilistic reachability properties. These provide rigor and superior expressiveness and can encode common performance indices for PID controllers~\cite{li2004cautocsd}, as shown in Section \ref{sec:implementation}. 

Other related work includes the Simplex architecture~\cite{Sha2001} where, whenever the plant is at risk of entering an unsafe state, the system switches from a high-performance advanced controller to a pre-certified (safe) baseline controller (with worse performance), leading to a potential trade-off between safety and performance. In our approach, performance and safety are instead equal cohorts in the synthesis process. 
Unlike Simplex, in the \emph{Control Barrier Function} (CBF) approach~\cite{AmesH14}, there is no baseline controller to fall back on: a CBF minimally perturbs a (possibly erroneous) control input to the plant so the plant remains in the safe region. As far as we know, neither Simplex nor CBFs have been designed with a stochastic plant model in mind. 

The controller synthesis problem under safety constraints (bounded STL properties in this case) is also considered in~\cite{Raman15}.  The main differences between this approach and ours is that they focus on Model Predictive rather than PID control, and their system model does not support stochastic parameters. There are a number of formal approaches (\eg, \cite{TronciTAC2017})
to control synthesis that consider the sample-and-hold schema typical of discrete-time controllers, but they do not yield PID controllers and cannot handle stochastic hybrid systems.
Verification of hybrid control systems with non-deterministic disturbances is considered in~\cite{mancini2013system} and solved through a combination of explicit model checking and simulation. However, unlike our method, it does not support controller synthesis and arbitrary probability distributions for the disturbances. 


There has been a sizable amount of work on tools for formal analysis of probabilistic reachability,
although they all have limitations that make them unsuitable for our approach. SiSAT~\cite{Sisat}
uses an SMT approach for probabilistic hybrid systems with discrete nondeterminism, while continuous 
nondeterminism is handled via Monte Carlo techniques only~\cite{SMCnd}; UPPAAL~\cite{UPPAALSMCtut} uses 
statistical model checking to analyze nonlinear stochastic hybrid automata; ProHVer~\cite{ZhangSRHH10} 
computes upper bounds for maximal reachability probabilities, but continuous random parameters are analyzed 
via discrete over-approximations~\cite{FraenzleHHWZ11}; U-Check~\cite{BortolussiMS15} enables 
parameter synthesis and statistical model checking of stochastic hybrid systems~\cite{BartocciBNS15}). 
However, this approach is based on Gaussian process emulation and optimisation, and provides only statistical 
guarantees and requires certain smoothness conditions on the satisfaction probability function.

Other approaches to solving SMT problems over nonlinear real arithmetic include the complete (over polynomials), yet computationally expensive, cylindrical algebraic decomposition method implemented in solvers like Z3~\cite{de2008z3}, as well as a recent method~\cite{cimatti2017invariant} based on the incremental linearization of nonlinear functions. However, none of these support ODEs and transcendental functions.

\section{Conclusions and Future Work}
\label{sec:conclusion}
The design of PID controllers for complex, safety-critical cyber-physical systems is challenging due to the hybrid, stochastic, and nonlinear dynamics they exhibit. Motivated by the need for high-assurance design techniques in this context, 
in this paper we presented a new method for the automated synthesis of PID controllers for stochastic hybrid systems from probabilistic reachability specifications. In particular, our approach can provide rigorous guarantees of safety and robustness for the resulting closed-loop system, while ensuring prescribed performance levels for the controller. We demonstrated the effectiveness of our approach on an artificial pancreas case study, for which safety and robustness guarantees are paramount.

As future work, we plan to study more advanced variants of the PID design such as nonlinear PID controllers, 
as well as investigate how common PID tuning heuristics can be integrated in our automated approach to speed up the search for suitable controllers.

\paragraph{\bf Acknowledgements}
Research supported in part by
EPSRC (UK) grant EP/N031962/1, 
FWF (Austria) S 11405-N23 (RiSE/SHiNE), 
AFOSR Grant FA9550-14-1-0261 
and NSF Grants IIS-1447549, 
CNS-1446832, 
CNS-1445770, 
CNS-1445770, 
CNS-1553273, CNS-1536086, CNS 1463722, and IIS-1460370. 

\bibliographystyle{abbrv}
\bibliography{memocode2017}

\appendix
\section{Proof of Proposition \ref{prop:formal}}\label{app:proof1}
\begin{proof}
Let $[\mathbf{k}_m]$ be the parameter box from which $\mathbf{k}^*$ was selected, and let $[P_m] = [P^\bot_m,P^\top_m]$ be the corresponding probability enclosure with minimal midpoint. In the best case, $[P_m]$ has also the least lower bound, implying that $P^* \in [P_m]$ and in turn that, $\mathbf{Pr}(\mathbf{k}^*) \leq P^* + \epsilon$. In the worst case, there exists another enclosure, $[P_M] = [P^\bot_M,P^\top_M]$ with a better lower bound than $[P_m]$, i.e, with $P^\bot_M < P^\bot_m$. This implies that the actual minimal probability might be in $[P_M]$ and not in $[P_m]$, which induces a worst-case probability error of $P^\bot_m - P^\bot_M$, leading to $\mathbf{Pr}(\mathbf{k}^*) \leq P^* + \epsilon + P^\bot_m - P^\bot_M$. Now note that $P^\bot_m - P^\bot_M$ cannot exceed the half length of $[P_M]$, because otherwise $[P_M]$ would be the enclosure with the lowest midpoint. It follows that $\mathbf{Pr}(\mathbf{k}^*) < P^* + \epsilon + \epsilon/2$. 
\end{proof}

\section{Gluco-regulatory ODE model}\label{app:model}

\begin{equation} \label{eq:model-odes}
\begin{split} 
\dot{Q_1}(t) &= -F_{01} - x_1Q_1 + k_{12}Q_2-F_R+EGP_0(1-x_3)+ 0.18U_G;\\
\dot{Q_2}(t) &= x_1Q_1-(k_{12} + x_2)Q_2; ~~~U_G(t) = \frac{D_GA_G}{0.18t_{maxG}^2}t e^\frac{-t}{t_{maxG}};\\
G(t) &=\frac{Q_1(t)}{V_G}; ~\dot{S_1}(t) = u(t) + u_b - \frac{S_1}{t_{maxI}}; ~\dot{S_2}(t) =\frac{S_1-S_2}{t_{maxI}};\\
\dot{I}(t) &= \frac{S_2}{t_{maxI}V_I}-k_e I; ~~~\dot{x_i}(t) =-k_{a_i}x_i+k_{b_i}I; ~~~(i = 1,2,3) \\
\end{split}
\end{equation}

The model consists of three subsystems:
\begin{itemize}
	\item {\em Glucose Subsystem}: it tracks the masses of glucose (in mmol) in the accessible ($Q_1(t)$) and non-accessible ($Q_2(t)$)
    compartments, $G(t)$ (mmol/L) represents the glucose concentration in plasma, $EGP_0$ (mmol/min) is the endogenous glucose production rate and
    $U_G(t)$ (mmol/min) defines the glucose absorption rate after consuming $D_G$ grams of carbohydrates. $D_G$ represents the main external disturbance of the system.
    \item {\em Insulin Subsystem}: it represents absorption of subcutaneously administered insulin. It is defined by 
    a two-compartment chain, $S_1(t)$ and $S_2(t)$ measured in U (units of insulin),  where $u(t)$ (U/min) is the administration of insulin computed by the 
    PID controller, $u_b$ (U/min) is the basal insulin infusion rate and $I(t)$ (U/L) indicates the insulin concentration in plasma.
    \item {\em Insulin Action Subsystem}: it models the action of insulin on glucose distribution/transport, $x_1(t)$,
    glucose disposal, $x_2(t)$, and endogenous glucose production, $x_3(t)$ (unitless).
\end{itemize}
The model parameters are given in Table
\ref{table:params}. 

\begin{table}
\centering
\begin{tabular}{|c|c||c|c||c|c|}
\hline
\bf{par} & \bf{value} & \bf{par} & \bf{value} & \bf{par} & \bf{value} \\
\hline
\hline
$w$ & 100 & $k_e$ & 0.138 & $k_{12}$ & 0.066 \\ 
$k_{a_1}$ & 0.006 & $k_{a_2}$ & 0.06 & $k_{a_3}$ & 0.03 \\ 
$k_{b_1}$ & 0.0034 & $k_{b_2}$ & 0.056 & $k_{b_3}$ & 0.024 \\ 
$t_{maxI}$ & 55 & $V_I$ & $0.12\cdot w$ & $V_G$ & $0.16\cdot w$ \\ 
$F_{01}$ & $0.0097\cdot w$ & $t_{maxG}$ & 40 & $F_R$ & 0 \\
$EGP_0$ & $0.0161\cdot w$ & $A_G$ & 0.8 & & \\
\hline
\end{tabular}
\vspace{2ex}
\caption{Parameter values for the glucose-insulin regulatory model. $w$ (kg) is the body weight.}
\label{table:params}
\end{table}

\end{document}